\documentclass[runningheads]{llncs}
\usepackage[T1]{fontenc}
\usepackage{graphicx}
\usepackage{amsmath}
\usepackage{amssymb}
\usepackage{microtype}
\usepackage{xcolor}
\newcommand{\rau}{\textsc{All-Distinct}}

\DeclareMathOperator*{\polylog}{polylog}
\newcommand{\tru}{\ensuremath{\textsc{true}}}

\begin{document}
\title{Average-Case Optimal Encodings and Efficient Worst-Case Indices for Element Distinctness Queries}
\titlerunning{Average-Case Optimal Encodings for Element Distinctness Queries}
% If the paper title is too long for the running head, you can set
% an abbreviated paper title here
%
\author{Philip Bille\inst{1}\orcidID{0000-0002-1120-5154} \and
Johannes Fischer\inst{2}\orcidID{0000-0002-3384-597X} \and
Inge Li G{\o}rtz\inst{1}\orcidID{0000-0002-8322-4952}
\and
Filippo Lari\inst{3}\orcidID{0009-0000-6817-6561}}
\authorrunning{P. Bille et al.}
% First names are abbreviated in the running head.
% If there are more than two authors, 'et al.' is used.
%
\institute{Department of Applied Mathematics and Computer Science (DTU Compute),
Technical University of Denmark, Lyngby, Denmark\\
\email{\{phbi, inge\}@dtu.dk}
\and
Department of Computer Science, Technical University of Dortmund, Germany\\
\email{johannes.fischer@cs.tu-dortmund.de}
\and
Department of Computer Science, University of Pisa, Italy\\
\email{filippo.lari@phd.unipi.it}}
\maketitle              % typeset the header of the contribution

\begin{abstract} 
We study the data structure version of the \emph{element distinctness problem}: preprocess an array of $n$ elements from an alphabet of size $\sigma$ to answer \textsc{All-Distinct} queries, asking whether a given range contains only distinct elements. We first focus on \emph{uniformly random arrays}: in the encoding model, where access to the input at query time is not allowed, we prove a lower bound on the expected space; for instance, the lower bound is $n$, $1.3627n$, $1.5153n$, $1.5824n$ bits for $\sigma = 2,3,4,5$, and approximately $n\sqrt{\pi/(2\sigma)}\,\log\sigma$ bits for $\sigma =\omega(1)$. We complement this by designing different average-case optimal encodings, supporting \textsc{All-Distinct} queries in worst-case time $O(1)$, $o(\log^{2}{\log{n}})$, or $O(\log\log{n})$ depending on $\sigma$, and $O(1)$ expected time for any $\sigma = \omega(1)$. We then switch to worst-case (non-random) arrays: in the indexing model, where access to the input is allowed, we prove a cell-probe space-time tradeoff lower bound showing that any index using $n/b$ bits must have $\Omega(b/\log{b})$ query time. We conclude by presenting a simple index almost matching this lower bound.

\end{abstract}

\section{Introduction}\label{sec:introduction}

Range queries are an extensively studied class of problems in data structures, asking to preprocess a data set so that, given a query range, some function of its elements can be computed efficiently. Classic examples include range minimum~\cite{DBLP:conf/stoc/GabowBT84,DBLP:journals/siamcomp/BerkmanV93,bender2000lca,DBLP:journals/jda/FerradaN17,DBLP:conf/wea/BaumstarkGHL17,DBLP:journals/siamcomp/FischerH11,DBLP:reference/algo/Fischer16,bille2025dynamic,FerraginaL25,DBLP:journals/tcs/JoS26,DBLP:conf/wea/CarmonaL26}, range median~\cite{DBLP:journals/njc/KrizancMS05,DBLP:conf/sofsem/Petersen08a,PETERSEN2009225,DBLP:journals/tcs/BrodalGJS11,DBLP:conf/isaac/BrodalJ09}, and range mode~\cite{DBLP:journals/ipl/PetersenG09,DBLP:journals/mst/ChanDLMW14} to name a few. Such problems are commonly studied in two settings~\cite {DBLP:reference/algo/Fischer16}: the
\emph{indexing model}, where the data structure may access the input, and the
\emph{encoding model}, where the input is no longer available after preprocessing and queries must be answered from the encoding alone.

Recently, Fischer \& Lari~\cite{FischerL26} introduced the data
structure version of the well-known \emph{element distinctness problem}: given an array, preprocess it to answer \textsc{All-Distinct} queries, which ask
whether a given range contains only distinct elements. When $A$ is the document
array of a text, such queries have interesting applications in information
retrieval~\cite{FischerL26}. In their work, they introduce the following sequences, which we also use extensively:

\begin{definition}\label{def:p_n_array}
Let $A$ be an array of $n$ elements from an alphabet of size $\sigma$. For every $0 \le i < n$, let $P_A[i] = \min\{\, j \le i \mid \rau_A(j,i) = \tru\}$ and $N_A[i] = \max\{\, j \ge i \mid \rau_A(i,j) = \tru\}$.
\end{definition}

%In particular, in the encoding model, they prove that the $P$-sequences (or symmetrically the $N$-sequences) characterize the equivalence classes of the \textsc{All-Distinct} problem.

Noting that the answer of $\textsc{All-Distinct}(l,r)$ is \textsc{True} if and only if $P_{A}[r] \le l$ (or symmetrically $N_{A}[l] \ge r$), their solutions are based on efficiently storing and accessing such sequences in the two aforementioned models. In particular, in the encoding model, they prove that any data structure that answers \textsc{All-Distinct} queries can be used to reconstruct the $P$-sequence (or symmetrically the $N$-sequence) of the input array. Based on this, they give an information-theoretic lower bound of $2n-O(\log{n})$ bits, which is further refined to $n\log{r_{\sigma}}-3\log{(\sigma+2)}+O(1)$ bits when the alphabet size $\sigma$ is constant, where $r_{\sigma} = 4\cos^{2}{(\pi/(\sigma+2))}$. Both lower bounds are matched within lower-order terms with an encoding supporting $O(1)$ time queries. In the indexing model, for every $1 \le b \le n$, they obtain an index using $O((n/b)\log{b})$ bits with $T_{dist}(b,\sigma)$ query time, where $T_{dist}(b,\sigma)$ is the time needed to solve an $O(b)$-sized element distinctness instance, see~\cite{VanDerHoogRR2025-ElementDistinctness} for a summary of known tradeoffs.

In this work, we extend the results of~\cite{FischerL26}. Specifically, in Sec.~\ref{sec:entropy}, we focus on refining their lower bounds in the encoding model to the average case, assuming that the instances are drawn uniformly at random. To this end, we use the result of Alon \& Orlitsky~\cite{AlonO94}, which lower-bounds the average code length of any one-to-one code for a discrete random variable $X$ by $H(X)-\log{(H(X)+1)}-\log{e}$, where $H(X)$ is the Shannon entropy~\footnote{Notice that the Shannon entropy alone lower-bounds the average length of prefix-free (or, more generally, uniquely decodable) codes. Here we consider arbitrary one-to-one binary encodings, for which the entropy bound does not directly apply, and Alon \& Orlitsky's result is needed.}. We therefore analyze the Shannon entropy of the $P$-sequences, obtaining a lower bound expressed in terms of two functions, $Q$ and $C$, defined below.
%To this end, we analyze the Shannon entropy of the $P$-sequences, obtaining a lower bound that is expressed in terms of two functions, $Q$ and $C$, defined below.

\begin{definition}\label{def:Q_function}
Let $p_k = \prod_{j=1}^{k}(1 - j/\sigma)$ for $k \ge 0$, where $p_0 = 1$ and $p_k = 0$ for $k \ge \sigma$.
\begin{equation*}
  Q(\sigma) := \sum_{k=0}^{\sigma-1} p_k,\quad C(\sigma) := \sum_{k=1}^{\sigma-1}\frac{k}{\sigma}\,p_k\log\frac{1}{1-k/\sigma}.
\end{equation*}
$Q(\sigma)$ is known as Ramanujan's Q-function~\cite{Flajolet95}.
\end{definition}

\begin{theorem}\label{th:p_sequence_entropy_sigma}
Let $A$ be an array of $n$ elements drawn independently and uniformly at random from $[1, \sigma]$, with $1 \le \sigma \le n$. Any encoding of $A$ supporting \textsc{All-Distinct} queries requires, in expectation, at least $H(P_A)-O(\log{n})$ bits, where
\begin{equation*}
    H(P_{A}) = \frac{\big(nQ(\sigma) - \sigma\big)\log\sigma}{\sigma}
    + nC(\sigma) + O(\sigma C(\sigma)) \quad \text{bits.}
\end{equation*}
In particular, as $\sigma \to \infty$,
\begin{equation*}
    H(P_{A}) \ge n\sqrt{\frac{\pi}{2\sigma}}\log\sigma
        - \frac{n\,\log\sigma}{3\sigma}
        - \log\sigma
        + O\!\left(\frac{n\,\log\sigma}{\sigma^{3/2}}\right) \quad \text{bits.}
\end{equation*}
\end{theorem}

As an immediate consequence, some values of the leading terms in the lower bound can be explicitly computed for $\sigma = O(1)$ (e.g., $2, 3, 4, 5$), giving $n$, $1.3627n$, $1.5153n$, and $1.5824n$ bits. When $\sigma=\omega(1)$ instead, for instance at $\sigma=n$, the lower bound is approximately $\sqrt{\pi n/2}\,\log n - O(\log{n})$ bits. In both cases, for random arrays, this surpasses the worst-case lower bounds~\cite{FischerL26}. Based on Thm.~\ref{th:p_sequence_entropy_sigma}, in Sec.~\ref{sec:enc_const_alphabet} and~\ref{sec:enc_gen_alphabet}, we design three average-case optimal encodings.
%providing different query time tradeoffs depending on $\sigma$.

\begin{theorem}\label{th:opt_encodings_all}
Let $A$ be an array of $n$ elements drawn independently and uniformly at random from $[1,\sigma]$, with $1 \le \sigma \le n$. There exist different encodings of $A$ supporting \textsc{All-Distinct} queries whose expected size matches the lower bound of Thm.~\ref{th:p_sequence_entropy_sigma} up to lower-order terms. These encodings provide the following tradeoffs:
\begin{itemize}
  \item worst-case query time
    \[
      \begin{cases}
        O(1)              & \sigma = O(\operatorname{polylog} n),\\
        o(\log^{2}{\log{n}}) & \sigma = n^{o(1)},\\
        O(\log{\log{n}})     & \sigma = \Theta(n^{\beta}),\ \text{for any constant}\ \beta \in (0,1],
      \end{cases}
    \]
    \item $O(1)$ expected query time, for every $\sigma = \omega(1)$
\end{itemize}
\end{theorem}

Moving to the indexing model, in Sec.~\ref{sec:cell_probe_lb}, we prove the following lower bound.

%Moving to the indexing model, in Sec.~\ref{sec:cell_probe_lb}, we prove the following space-time tradeoff lower bound.

\begin{theorem}\label{th:cell_probe_lb}
%Given $n$ and $b$ such that $1< b < n$, let $A$ be an array of $n$ elements from an alphabet of size $b+1\le \sigma \le n$. Any indexing data structure solving \textsc{All-Distinct} queries using $n/b$ bits of space on top of $A$ must have $\Omega(b/\log{b})$ query time.
Given $n$ and $b$ such that $1< b < n$, let $A$ be an array of $n$ elements from an alphabet of size $b+1\le \sigma \le n$. Any indexing data structure solving \textsc{All-Distinct} queries using $n/b$ bits of space must have $\Omega(b/\log{b})$ query time.
\end{theorem}

In Sec.~\ref{sec:new_index}, we conclude by designing a very simple index using $O(n/b)$ bits (i.e., shaving an $O(\log{b})$ factor from the one of~\cite{FischerL26}) having the same $T_{dist}(b,\sigma)$ query time.
Here, $T_{dist}(b,\sigma)$ denotes the time to solve an online (non-preprocessed) instance of the element distinctness problem.
As an example, if an optimal comparison-based sorting algorithm is used to solve the $O(b)$-sized element distinctness instance, $T_{dist}(b,\sigma) = O(b\log{b})$, and thus it is only an $O(\log^{2}{b})$ factor away from the optimal query time within that space. We leave open the problem of either designing an index that matches Thm.~\ref{th:cell_probe_lb} or tightening the lower bound. Lastly, we note that all our encoding results hold in the \emph{transdichotomous word-RAM model}, where the word size $w$ satisfies $w = \Theta(\log{n})$, and arithmetic as well as bitwise operations on $w$ bits are performed in $O(1)$ time. Additionally, we assume that logarithms hide their ceiling and thus return integers, and we adopt the standard entropy convention $0\log(1/0) = 0$~\cite{cover1999elements}. 

%(equivalently, $\lim_{x\to 0^{+}} x\log(1/x) = 0$)~\cite{cover1999elements}.

\section{A Space Lower Bound for Encoding Random Instances}\label{sec:entropy}

%We begin by introducing some useful results (proofs in Appendix~\ref{app:missing_proofs_entropy}) that are extensively used in the following.
We begin with some useful results that are used in the following (see Appendix~\ref{app:missing_proofs_entropy}).
\begin{lemma}\label{lem:identities}
The following identities hold for every $\sigma \ge 1$:
\begin{align}
    p_{k-1} - p_k &\;=\; \tfrac{k}{\sigma}\,p_{k-1},
        \qquad 1 \le k \le \sigma, \label{eq:id_basic} \\[4pt]
    \sum_{k=0}^{\sigma-1} k\,p_k &\;=\; \sigma - Q(\sigma). \label{eq:id_kpk}
\end{align}
\end{lemma}

\begin{lemma}\label{lem:p_i_q}
Let $A$ be an array of $n$ elements drawn independently and uniformly at random from $[1,\sigma]$. For every $1 \le i \le n-1$ and $0 \le q \le i-1$,
\begin{equation*}
    \mathbb{P}(P_{A}[i-1] = q) \;=\;
    \begin{cases}
        p_{i-1}                 & q = 0,\; i \le \sigma, \\[2pt]
        \tfrac{(i-q)}{\sigma}p_{(i-q)-1}     & q \ge 1,\; 1 \le i-q \le \sigma, \\[2pt]
        0                       & \text{otherwise.}
    \end{cases}
\end{equation*}
\end{lemma}

\begin{lemma}\label{lem:p_i_cond}
Let $A$ be an array of $n$ elements drawn independently and uniformly at random from $[1,\sigma]$. For every $1 \le i \le n-1$ and every $0 \le q \le i-1$ with $i - q \le \sigma$,
\begin{equation*}
    \mathbb{P}(P_{A}[i] = p \mid P_{A}[i-1] = q) \;=\;
    \begin{cases}
        1 - (i - q)/\sigma & p = q, \\[2pt]
        1/\sigma           & q < p \le i, \\[2pt]
        0                  & \text{otherwise.}
    \end{cases}
\end{equation*}
\end{lemma}

Given an array $A$, any encoding that answers \textsc{All-Distinct} queries on $A$ can be used to reconstruct $P_{A}$~\cite{FischerL26}. By the result of Alon \& Orlitsky~\cite{AlonO94}, the expected size of such an encoding is at least $H(P_{A})-\log(H(P_A)+1)-\log e$ bits. Thus, it suffices to compute $H(P_{A})$, assuming that each element of $A$ is drawn independently and uniformly at random from $[1,\sigma]$. 
%Given an array $A$, any encoding answering \textsc{All-Distinct} queries on $A$ can be used to reconstruct $P_{A}$~\cite{FischerL26}. To obtain a lower bound on the expected space usage of such encodings, we compute the Shannon entropy $H(P_{A})$ of $P_{A}$, assuming that each element of $A$ is drawn independently and uniformly at random from $[1,\sigma]$. 
%To not overload the notation, we drop the subscript on $P_{A}$ when it is clear from the context.
\begin{align}
    H(P_{A}) &= H(P_{A}[0],\dots,P_{A}[n-1]) \nonumber
    \quad \text{(chain rule, see~\cite[Sec.~2.2]{cover1999elements})} \\[2pt]
    &= H(P_{A}[0]) + \sum_{i=1}^{n-1}H(P_{A}[i]\, \mid\, P_{A}[0],\dots,P_{A}[i-1]). \label{eq:cond_entropy}
\end{align}

Consider $H(P_{A}[i] \mid P_{A}[0],\dots,P_{A}[i-1])$. Conditional on $P_{A}[i-1] = q$, the values $A[q,i-1]$ form a tuple of $i - q$ distinct symbols from $[1,\sigma]$, while $A[i]$ is drawn uniformly and independently from $[1,\sigma]$. Since $P_{A}[i]$ is a function of $A[q,i]$ and, by Lemma~\ref{lem:p_i_cond}, its conditional distribution given $P_{A}[i-1] = q$ depends only on $i-q$ and $\sigma$, we have $H(P_{A}[i] \mid P_{A}[0],\dots,P_{A}[i-1]) = H(P_{A}[i] \mid P_{A}[i-1])$. Moreover, since $P_{A}[0] = 0$ with probability $1$, $H(P_{A}[0]) = 0$. Therefore:
%Consider $H(P_{A}[i] \mid P_{A}[0],\dots,P_{A}[i-1])$. Conditional on $P_{A}[i-1] = q$, the suffix $A[q,i-1]$ is a uniformly random tuple of $i - q$ distinct values from $[1,\sigma]$, independent of $A[0,q-1]$. Since $A[i]$ is independent of $A[0,i-1]$ and $P_{A}[i]$ is a function of $A[q,i]$ alone, the conditional distribution of $P_{A}[i]$ given the full prefix depends only on $P_{A}[i-1]$, thus $H(P_{A}[i] \mid P_{A}[0],\dots,P_{A}[i-1]) = H(P_{A}[i] \mid P_{A}[i-1])$. Moreover, since $P_{A}[0] = 0$ with probability $1$, $H(P_{A}[0]) = 0$. Therefore:
\begin{align}
    H(P_{A})
    &= \sum_{i=1}^{n-1} H(P_{A}[i] \mid P_{A}[i-1])
    \quad\text{(conditional entropy, see~\cite[Sec.~2.2]{cover1999elements})}
    \label{eq:entropy} \\[2pt]
    &= \sum_{i=1}^{n-1}\sum_{q=0}^{i-1}\sum_{p=q}^{i}
    \mathbb{P}(P_{A}[i] = p, P_{A}[i-1]=q)
    \nonumber\\[-1pt]
    &\qquad\qquad {}\cdot
    \log\big(1/\mathbb{P}(P_{A}[i]=p \mid P_{A}[i-1]=q)\big)
    \nonumber \\[2pt]
    &= \sum_{i=1}^{n-1}\sum_{q=0}^{i-1}\sum_{p=q}^{i}
    \mathbb{P}(P_{A}[i]=p \mid P_{A}[i-1]=q)
    \mathbb{P}(P_{A}[i-1] = q)
    \nonumber\\[-1pt]
    &\qquad\qquad {}\cdot
    \log\big(1/\mathbb{P}(P_{A}[i]=p \mid P_{A}[i-1]=q)\big).
    \nonumber
\end{align}

By Lemma~\ref{lem:p_i_cond}, when $q+1\le p \le i$ we have $\mathbb{P}(P_{A}[i]=p \mid P_{A}[i-1]=q) = 1/\sigma$, hence $\log(1/\mathbb{P}(P_{A}[i]=p \mid P_{A}[i-1]=q)) = \log\sigma$; the remaining term $p = q$ gives $\mathbb{P}(P_{A}[i]=q \mid P_{A}[i-1]=q) = 1 - (i-q)/\sigma$, thus $\log(1/\mathbb{P}(P_{A}[i]=p \mid P_{A}[i-1]=q)) = \log{(1/(1-(i-q)/\sigma)}$. Separating these two regimes in Eq.~\ref{eq:entropy} gives
\begin{equation}\label{eq:H_split}
\begin{aligned}
H(P_{A})
  &= \frac{\log \sigma}{\sigma}
     \sum_{i=1}^{n-1} S_i
     + \sum_{i=1}^{n-1} D_i,\quad
S_i := \sum_{q=0}^{i-1}
     \mathbb{P}(P_{A}[i-1] = q)\,(i-q), \\[0.5ex]
D_i
  &:= \sum_{q=0}^{i-1}
     \mathbb{P}(P_{A}[i-1] = q)
     \left(1 - \frac{i-q}{\sigma}\right)
     \log\!\left(
       \frac{1}{1-(i-q)/\sigma}
     \right).
\end{aligned}
\end{equation}

We proceed by evaluating the two sums separately. Starting from $S_i$, we set $k = i-q$ and distinguish between two cases:

\begin{itemize}
\item \emph{Case 1: $1 \le i \le \sigma$.} Applying Lemma~\ref{lem:p_i_q} (the $q = 0$ and $q \ge 1$ branches):
\begin{equation*}
    S_i = ip_{i-1} + \sum_{k=1}^{i-1}\frac{k^2}{\sigma}p_{k-1}.
\end{equation*}
By Eq.~\ref{eq:id_basic} of Lemma~\ref{lem:identities}, $(k/\sigma)\,p_{k-1} = p_{k-1} - p_k$, so the sum equals
$\sum_{k=1}^{i-1} k(p_{k-1} - p_k)$, which is equal to $\sum_{k=0}^{i-2} p_k - (i-1)p_{i-1}$. Therefore:
\begin{equation}\label{eq:Si_case1}
    S_i = ip_{i-1} + \sum_{k=0}^{i-2} p_k - (i-1)p_{i-1} = \sum_{k=0}^{i-1} p_k.
\end{equation}

\item \emph{Case 2: $i > \sigma$.} Applying Lemma~\ref{lem:p_i_q} (the $q \ge 1$ branch): $S_i = \sum_{k=1}^{\sigma}\frac{k^2}{\sigma}\,p_{k-1}$. Again by Eq.~\ref{eq:id_basic}, the sum equals
$\sum_{k=1}^{\sigma} k\,(p_{k-1} - p_k)$, which gives the following:
\begin{equation}\label{eq:Si_case2}
    S_i = \sum_{k=1}^{\sigma} k(p_{k-1} - p_k) = \sum_{k=0}^{\sigma-1} p_k - \sigma p_\sigma
    = Q(\sigma)
\end{equation}
where the last equality uses $p_\sigma = 0$ (recall Def.~\ref{def:Q_function})
\end{itemize}

Splitting the first sum of Eq.~\ref{eq:H_split} at $i = \sigma$ and using Eq.~\ref{eq:Si_case1} and Eq.~\ref{eq:Si_case2}:
\begin{align}
    \frac{\log\sigma}{\sigma}\sum_{i=1}^{n-1}S_i
    &= \frac{\log\sigma}{\sigma}
    \left(
        \sum_{i=1}^{\sigma}\sum_{k=0}^{i-1}p_k
        + (n-\sigma-1)Q(\sigma)
    \right)
    \nonumber\\
    &= \frac{\log\sigma}{\sigma}
    \left(
        \sum_{k=0}^{\sigma-1}(\sigma-k)p_k
        + (n-\sigma-1)Q(\sigma)
    \right)
    \nonumber\\
    &= \frac{\bigl(nQ(\sigma)-\sigma\bigr)\log\sigma}{\sigma}.
    \label{eq:S_i}
\end{align}
where the last step uses Eq.~\ref{eq:id_kpk} of Lemma~\ref{lem:identities}.

Moving to the second sum of Eq.~\ref{eq:H_split}, i.e., the one running over $D_i$, we proceed similarly. Setting $k = i-q$, we distinguish again between two cases:

\begin{itemize}
\item \emph{Case 1: $1 \le i \le \sigma$.} Applying Lemma~\ref{lem:p_i_q} (the $q = 0$ and $q \ge 1$ branches) and noting that from Eq.~\ref{eq:id_basic} of Lemma~\ref{lem:identities} follows $p_{k-1}(1-k/\sigma) = p_k$:
\begin{equation*}
    D_i = p_i\log\frac{1}{1-i/\sigma} + \sum_{k=1}^{i-1}\frac{k}{\sigma}p_k\log\frac{1}{1-k/\sigma}
\end{equation*}
\item \emph{Case 2: $i > \sigma$.} Applying Lemma~\ref{lem:p_i_q} (the $q \ge 1$ branches) and using the same observation as in the previous case:
\begin{equation*}
    D_i = \sum_{k=1}^{\sigma-1}\frac{k}{\sigma}p_k\log\frac{1}{1-k/\sigma} := C(\sigma)
\end{equation*}
\end{itemize}
Splitting Eq.~\ref{eq:H_split} based on these two cases, we obtain:
\begin{align}
    \sum_{i=1}^{n-1} D_i &= \sum_{i=1}^{\sigma}\left(p_i\log\frac{1}{1-i/\sigma} + \sum_{k=1}^{i-1}\frac{k}{\sigma}p_k\log\frac{1}{1-k/\sigma}\right)+(n-\sigma-1)C(\sigma) \nonumber\\
    &\le 2\sigma C(\sigma) + (n-\sigma-1)C(\sigma) \nonumber\\
    &= nC(\sigma) + (\sigma-1)C(\sigma)\label{eq:d_i}
\end{align}

Where the inequality follows after noting that $\sum_{i=1}^{\sigma}p_i\log{(1/(1-i/\sigma))} \le \sigma\sum_{i=1}^{\sigma-1}\tfrac{i}{\sigma}p_i\log{(1/(1-i/\sigma))} = \sigma C(\sigma)$ (recall $p_{\sigma} = 0$ and we assumed $0\log{1/0} = 0$), and that $\sum_{k=1}^{i-1}\tfrac{k}{\sigma}p_k\log{(1/(1-k/\sigma))} \le C(\sigma)$. 

Combining Eq.~\ref{eq:S_i} and~\ref{eq:d_i} inside Eq.~\ref{eq:H_split} gives

\begin{equation}\label{eq:exact_entropy}
    H(P_A) = \frac{\big(nQ(\sigma) - \sigma\big)\log\sigma}{\sigma}
    + nC(\sigma) + O(\sigma C(\sigma)) \quad \text{bits.}
\end{equation}

Note that, when $\sigma = O(1)$, the coefficient of the leading term in Eq.~\ref{eq:exact_entropy} can be computed exactly, giving the values shown in Sec~\ref{sec:introduction}. Moving to larger alphabet sizes that scale together with the input size, we can obtain a lower bound on the entropy by dropping the positive contribution of Eq.~\ref{eq:d_i}, thus obtaining $H(P_{A}) \ge ((nQ(\sigma) - \sigma)\log{\sigma})/\sigma$ bits. Using Ramanujan's Q-function asymptotic $Q(\sigma) = \sqrt{\pi\sigma/2} - 1/3 + O(\sigma^{-1/2})$ (see~\cite[\S 1.2.11.3, Eq.~25]{Knuth97}) gives

\begin{equation}\label{eq:asympt_entropy}
    H(P_{A}) \ge n\sqrt{\frac{\pi}{2\sigma}}\log\sigma
        - \frac{n\,\log\sigma}{3\sigma}
        - \log\sigma
        + O\!\left(\frac{n\,\log\sigma}{\sigma^{3/2}}\right) \quad \text{bits.}
\end{equation}
Finally, the lower bound on the expected encoding size is $H(P_A)-\log(H(P_A)+1)-\log e$~\cite{AlonO94}. Moreover, $H(P_A)$ is upper-bounded by the worst-case entropy of the
$P$-sequences, which is $2n-O(\log n)$ bits~\cite{FischerL26}. Therefore, $\log(H(P_A)+1)+\log e=O(\log n)$, and this concludes the proof of Thm.~\ref{th:p_sequence_entropy_sigma}.

%In the next sections, we design encodings whose space usage is average-case optimal across all alphabet regimes. This is interesting when $\sigma \to \infty$ as the encodings match the leading term of Eq.~\ref{eq:asympt_entropy} up to lower-order terms. Hence, discarding the (nonnegative) contribution of Eq.~\ref{eq:d_i} costs nothing at the leading term, and the bound is tight in its dominant term.

%We use Taylor's expansion of $\ln(1/(1-x))$, to express $C(\sigma)$ in terms of $Q(\sigma)$.
%\jo{$Q$ does not appear in the following term.}
%\begin{equation*}
%    C(\sigma) := \sum_{k=1}^{\sigma-1}\tfrac{k}{\sigma}\,p_k\log{(1/(1-k/\sigma))}  = \frac{1}{\sigma^2 \ln{2}}\sum_{k=1}^{\sigma-1}k^{2}p_{k}\left(1+O(\sigma^{-1/2})\right)
%\end{equation*}

%By Eq.~\ref{eq:id_kpk} of Lemma~\ref{lem:identities}, $\sum_{k=1}^{\sigma-1}k^{2}p_{k} = \sum_{k=1}^{\sigma-1}k \sigma (p_{k-1}-p_{k}) = (\sigma+1)Q(\sigma)-2\sigma$, and thus $C(\sigma) = O(Q(\sigma)/\sigma)$. Finally, using Ramanujan's Q-function asymptotic $Q(\sigma) = \sqrt{\pi\sigma/2} - 1/3 + O(\sigma^{-1/2})$ (see~\cite[\S 1.2.11.3, Eq.~25]{Knuth97}), we obtain

%\begin{equation}\label{eq:asympt_entropy}
%    H(P) = n\,\sqrt{\frac{\pi}{2\sigma}}\,\log\sigma
%        \;-\; \frac{n\,\log\sigma}{3\sigma}
%        \;-\; \log\sigma
%        \;+\; O\!\left(\frac{n\,\log\sigma}{\sigma^{3/2}}\right).
%\end{equation}
%Overall, combining Eq.~\ref{eq:exact_entropy} and~\ref{eq:asympt_entropy} proves Thm.~\ref{th:p_sequence_entropy_sigma}.

\section{An Optimal Encoding for Constant Alphabet Sizes}\label{sec:enc_const_alphabet}
We use a folklore variant of Huffman coding known as block Huffman coding. This allows us to approach the entropy of the underlying sequence arbitrarily closely. The tradeoff is that we must build one or several potentially large Huffman trees. However, as we demonstrate below, this cost remains manageable when $\sigma = O(1)$ and the block size is chosen carefully.

Given a block size $b \ge 1$ (to be fixed later), and assuming for simplicity that $b \mid n$, we partition $P_A$ into $n/b$ consecutive blocks, each of size $b$. For $0\le i< n/b$, let $B_i=P_{A}[i\,b],\dots,P_{A}[(i+1)\,b-1]$ denote the content of the $i$-th block. We treat each block as a single symbol, with probability determined by the uniform model on $A$ given by Lemma~\ref{lem:p_i_cond} rather than by the empirical frequencies of $P_A$, conditioned on the last value of the preceding block (if any):
%with a probability distribution conditioned on the last value of the previous block (if it exists): 
\begin{equation*}
    \mathbb{P}(B_i \mid P_{A}[ib-1]) = \prod_{k=ib}^{(i+1)b-1} \mathbb{P}(P_{A}[k]\mid P_{A}[k-1])
\end{equation*}
This formula follows from the same machinery of the proof of Thm.~\ref{th:p_sequence_entropy_sigma} and that $P_{A}[0] = 0$ deterministically, which applies to the first block. 

We then build a separate Huffman tree for each possible preceding value. Since $i-\sigma+1 \le P_{A}[i] \le i$, we only require $\sigma$ different trees, selected according to the value of $K_i = i-P_{A}[i]$ (assume $K_0 = 0$ for the first block).

Let $\ell_i$ denote the expected codeword length of the $i$-th block. By the Huffman coding bound~\cite[Sec.~5.6]{cover1999elements}, and the same argument we used for Eq.~\ref{eq:cond_entropy} in Sec.~\ref{sec:entropy}:
\begin{equation*}
    \ell_i < H(B_i \mid P_{A}[ib-1]) + 1 = \left(\sum_{k=ib}^{(i+1)b-1}H(P_{A}[k]\mid P_{A}[k-1])\right) + 1
\end{equation*}

Summing over all blocks, because of Eq.~\ref{eq:entropy}, the expected size in bits of the entropy-coded sequence can be bounded as follows:
\begin{equation}\label{eq:entropy_enc_const_sigma}
    \sum_{i=0}^{n/b-1}\ell_i < \sum_{i=0}^{n/b-1}\left(\left(\sum_{k=ib}^{(i+1)b-1}H(P_{A}[k]\mid P_{A}[k-1])\right)+1\right) = H(P_{A}) + n/b
\end{equation}

To randomly access any position of such a sequence (and thereby answer \textsc{All-Distinct} queries), we store the following additional data structures:
\begin{itemize}
    \item A select data structure $S_{P_A}$ supporting $O(1)$ time queries, built over the starting positions of the blocks in the encoded sequence. This requires $(n/b)\log{b} + O(n/b)$ bits, e.g., using Elias-Fano codes~\cite[§3.4.3]{Navarro16book}.
    
    \item For each block $1\le i < n/b$, the value $K_{ib-1} \in \{0,\dots,\sigma-1\}$ needed to select the correct Huffman tree (or table, see later), stored in an array $B_{P_A}$ using $(n/b)\log\sigma = O(n/b)$ bits since $\sigma = O(1)$.
    
    \item A precomputed lookup table $T_K[\,\cdot\,][\,\cdot\,]$ for each $K \in \{0,\dots,\sigma-1\}$. Given a codeword $c$, the entry $T_{K}[c][j] \in \{0,\dots,\sigma-1\}$ returns the relative value $K_{ib+j}$ of the $j$-th symbol of the block whose preceding state is $K$ and whose codeword is $c$; the decoded value is then recovered as $P_{A}[ib+j] = (ib+j) - K_{ib+j}$. These tables occupy $O(\sigma(2^{\ell_{max}} b\log{\sigma})) = O(2^{\ell_{max}} b)$ bits (recall $\sigma=O(1)$ here), where $\ell_{max}$ is the maximum codeword length.
\end{itemize}

Considering the size of each lookup table, it is well-known that the maximum codeword length produced by Huffman coding is bounded by $\log_{\phi}{(1/p_{\text{min}})}$, where $\phi \approx 1.6180$ is the \emph{golden ratio} and $p_{\text{min}}$ is the smallest probability in the underlying distribution~\cite{Navarro16book}. In our case, considering only blocks of positive probability, every conditional probability $\mathbb{P}(P_{A}[k] = p\mid P_{A}[k-1] = q)$ is at least $1/\sigma$. By Lemma~\ref{lem:p_i_cond}, it equals $1/\sigma$ when $p>q$, and $1-(k-q)/\sigma$ when $p=q$. The latter case is positive only when $k-q\leq\sigma-1$, and hence its probability is also at least $1/\sigma$. Therefore, $p_{\text{min}} \ge 1/\sigma^{b}$, and consequently $\ell_{max} \le b\log_{\phi}\sigma \approx 1.4404b\log{\sigma}$ bits. Setting $b = \tfrac{\log{\phi}}{2}\log_{\sigma}{n} \approx 0.3471\log_{\sigma}{n}$, each lookup table uses $O(\sqrt{n}\log{n})$ bits. Under the same choice of $b$, the space usage of the auxiliary data structures is dominated by the $O((n\log\log{n})/\log{n})$ bits of $S_{P_A}$. Therefore, the overall expected space usage is $((nQ(\sigma) -\sigma)\log\sigma)/\sigma + nC(\sigma) + O((n\log{\log{n}})/\log{n})$ bits, matching the lower bound of Thm.~\ref{th:p_sequence_entropy_sigma} up to lower-order terms whenever $\sigma = O(1)$. Some values of the leading term are reported in Sec.~\ref{sec:introduction}.

%Combining this with Eq.~\ref{eq:entropy_enc_const_sigma} and Thm.~\ref{th:p_sequence_entropy_sigma}
%\begin{equation*}
%    \frac{\big(n\,Q(\sigma) - \sigma\big)\,\log\sigma}{\sigma}
%    \;+\; n\,C(\sigma) + O\!\left(\frac{n\log{\log{n}}}{\log{n}}\right) \quad\text{bits.}
%\end{equation*}

%For instance, for $\sigma \in \{2, 3, 4, 5\}$ the encoding uses on average roughly $n$, $1.3627n$, $1.5153n$, and $1.5824n$ bits, respectively. 

%Let $k = \lfloor i/b\rfloor$ be the block index. Read the previous state $K = B_{P_A}[k]$ and the block's starting position $s = S_{P_A}[k]$, all in $O(1)$ time.

Accessing $P_{A}[i]$ (and thus answering an \textsc{All-Distinct} query) is straightforward. Let $k = \lfloor i/b\rfloor$, we read the previous state $K = B_{P_A}[k]$ and the block's starting position $s = S_{P_A}[k]$, all in $O(1)$ time. Extract the $b\log_{\phi}\sigma = \tfrac{1}{2}\log n$ bits starting at position $s$ into a single machine word~\footnote{Padding with $\tfrac{1}{2}\log{n}$ bits at the end ensures that we never read outside of the encoded sequence.}. Let $d$ be its numerical value, then $T_K[d][\,i \bmod b\,] = K_i$, and thus we return $P_{A}[i] = i - K_i$ in $O(1)$ time.

This concludes the proof of the first branch of Thm.~\ref{th:opt_encodings_all} for $\sigma = O(1)$. The same encoding also extends to $\sigma = O(\polylog{n})$; however, in the next section we present an alternative data structure for this regime with a smaller lower-order term, namely, $O(n/\log n)$ instead of $O((n\log\log n)/\log n)$.
%Finally, although this encoding can match the lower bound of Thm.~\ref{th:p_sequence_entropy_sigma} even for $\sigma = O(\polylog{n})$, in the next section we present an alternative data structure with a smaller lower-order term, namely, $O(n/\log{n})$ instead of $O((n\log{\log{n}})/\log{n})$. This concludes the proof of the first branch of Thm.~\ref{th:opt_encodings_all} for $\sigma = O(1)$.

\section{An Optimal Encoding for General Alphabet Sizes}\label{sec:enc_gen_alphabet}

We begin by proving the following property of the $P$-sequences.

\begin{lemma}\label{lem:average_runs}
Let $A$ be an array of $n$ elements drawn independently and uniformly at
random from $[1, \sigma]$. Let $r$ denote the number of runs of identical elements in $P_A$. Assuming $\sigma \to \infty$, for large enough $n$, $\mathbb{E}[r] = n\sqrt{\pi/(2\sigma)}(1+o(1))$. 
\end{lemma}

\begin{proof}
For every $1 \le i < n$, a run of identical elements in $P_{A}$ begins whenever $P_{A}[i] > P_{A}[i-1]$. Let $X_i = \mathbf{1}\{P_{A}[i] > P_{A}[i-1]\}$ denote the indicator variable for such an event, then $r = 1 + \sum_{i=1}^{n-1}X_i$, and taking the expected value gives:
\begin{align*}
    \mathbb{E}[r] &= 1+\sum_{i=1}^{n-1}\mathbb{P}(P_{A}[i] > P_{A}[i-1]) \quad \text{(Conditioning on $P_{A}[i-1]$)} \\[2pt]
    &= 1+\sum_{i=1}^{n-1}\sum_{q=0}^{i-1}\mathbb{P}(P_{A}[i] > q\, |\, P_{A}[i-1] = q)\mathbb{P}(P_{A}[i-1] = q) \quad \text{(Lemma~\ref{lem:p_i_cond})} \\[2pt]
    &= 1+\sum_{i=1}^{n-1}\sum_{q=0}^{i-1}\mathbb{P}(P_{A}[i-1] = q)\frac{i-q}{\sigma} \quad \text{(Eq.~\ref{eq:H_split} proof of Thm.~\ref{th:p_sequence_entropy_sigma})} \\[2pt]
    &= 1+\frac{nQ(\sigma)-\sigma}{\sigma} \quad \text{(Ramanujan's asymptotic (see~\cite[\S 1.2.11.3, Eq.~25]{Knuth97})} \\[2pt]
    &= n\sqrt{\frac{\pi}{2\sigma}}(1+o(1))
\end{align*}
\end{proof}

Using Lemma~\ref{lem:average_runs}, we design two encodings matching the lower bound of Thm.~\ref{th:p_sequence_entropy_sigma}. The first one gives $O(1)$ expected query time for every $\sigma = \omega(1)$. The second one gives different worst-case query time tradeoffs depending on the alphabet size.

\paragraph{Constant Expected Query Time.}

We split $P_A$ into blocks of size $\sqrt{\sigma}$, assuming for simplicity that $\sqrt{\sigma} \mid n$. The resulting representation consists of three components. $R_{P_A}$ stores the concatenation of the starting position of every run inside a block, each as an offset relative to the start of the block, and thus encoded in $\log{\sqrt{\sigma}}$ bits. $S_{P_A}$ stores the $P_A$ value at the beginning of each run, and $C_{P_A}[i]$ stores the number of runs beginning before block $i$. We encode both $S_{P_A}$ and $C_{P_A}$ using Elias--Fano codes~\cite[\S3.4.3]{Navarro16book}. The overall expected space usage in bits is
\begin{equation*}
    \mathbb{E}\left[r\log{\frac{n}{r}}+r\log{\sqrt{\sigma}}+\frac{n}{\sqrt{\sigma}}\log{\frac{r\sqrt{\sigma}}{n}}+O\!\left(r+\frac{n}{\sqrt{\sigma}}\right)\right] \le n\sqrt{\frac{\pi}{2\sigma}}\log\sigma
        + O\!\left(\frac{n}{\sqrt{\sigma}}\right)
\end{equation*}
Each term in the expectation is a concave function of $r$: $r\log(n/r)$ is concave, $r\log\sigma$ is linear, $(n/\sqrt{\sigma})\log(r\sqrt{\sigma}/n)$ is concave in $r$, and the $O(r)$ and $O(n/\sqrt{\sigma})$ terms are linear or constant in $r$. By Jensen's inequality, $\mathbb{E}[f(r)] \le f(\mathbb{E}[r])$ for each concave term, and $\mathbb{E}[r] = n\sqrt{\pi/(2\sigma)}(1+o(1))$ by Lemma~\ref{lem:average_runs}. Substituting gives the stated bound. The space then matches the lower bound of Thm.~\ref{th:p_sequence_entropy_sigma} up to lower-order terms as long as $\sigma=\omega(1)$. Accessing $P_{A}[i]$ is straightforward. Let $j = \lfloor i/\sqrt{\sigma}\rfloor$ and $s = i \bmod \sqrt{\sigma}$, we compute $t$, the number of offsets in $R_{P_A}[j]$ smaller or equal than $s$, by a linear scan. The index of the run containing position $i$ is then $k = C_{P_{A}}[j]+t$, and $P_{A}[i] = S_{P_A}[k]$. From Lemma~\ref{lem:average_runs}, the average number of runs per block is $\sqrt{\pi/2} = O(1)$, therefore, accessing $P_A$ and thus answering an \textsc{All-Distinct} query takes $O(1)$ expected time.

\paragraph{Worst-Case Query Time.}
We store the starting position of each run of identical elements in $P_A$ along with their corresponding values. Let these two sequences be $R_{P_A}$ and $S_{P_A}$ respectively. Accessing a given value of $P_A$ at position $i$ requires a predecessor operation on $R_{P_A}$ to locate the run containing $i$, and then accessing the given position in $S_{P_A}$. One possibility is then to store both $R_{P_A}$ and $S_{P_A}$ using a Fully-Indexable Dictionary (FID), providing both the predecessor and access operations. The expected space usage of this encoding is
\begin{equation*}
    2\mathbb{E}\left[\log\binom{n}{r} + R(n,r,t)\right] \le n\sqrt{\frac{\pi}{2\sigma}}\log\sigma
     + O\!\left(\frac{n}{\sqrt{\sigma}}\right) + 2\mathbb{E}\left[R(n,r,t)\right]\quad\text{bits,}
\end{equation*}
where the upper bound follows by first applying the well-known inequality $\log\binom{n}{r}\leq r\log(ne/r)$, and then Jensen's inequality together with Lemma~\ref{lem:average_runs}. Here, $R(n,r,t)$ denotes the redundancy of the FID over the information-theoretic lower bound for answering queries in $O(t)$ time while storing a sequence of $r$ elements from a universe of size $n$. Noting that $O(n/\sqrt{\sigma}) = o((n/\sqrt{\sigma})\log{\sigma})$ for any $\sigma=\omega(1)$, this would match the lower bound of Thm.~\ref{th:p_sequence_entropy_sigma} up to lower-order terms as long as $\mathbb{E}[R(n,r,t)] = o((n/\sqrt{\sigma})\log{\sigma})$. It is well known that, depending on the interplay between the universe size $n$, the number of elements $r$, and the query time $t$, FIDs give different space-time tradeoffs. Since the expected value of $r$ depends on the alphabet size $\sigma$ and the sequence length $n$, we distinguish between the following regimes, matching the lower bound of Thm.~\ref{th:p_sequence_entropy_sigma} on expectation while delivering different worst-case query times (see  Appendix~\ref{app:missing_proofs_enc} for the full details).

%hence we focus on this redundancy term.

\begin{itemize}
    \item $\sigma = O(\polylog{n})$. We use Pătrașcu's FID~\cite{Patrascu08}, which supports predecessor and access in $O(t)$ time for any $t \ge 1$. The resulting encoding can be made to have $O(1)$ worst-case query time.

    \item $\sigma = n^{o(1)}$. We use Gupta et al.~\cite[Thm.~2]{GuptaHSV07} FID supporting access and predecessor in $O(\log{\log{r}})$ and $AT(n,r)$ time, where $AT(n,r)$ is Andersson and Thorup's predecessor data structure time bound~\cite{AnderssonT00,GuptaHSV07}. Our encoding has $o(\log^{2}{\log{n}})$ worst-case query time.
    %\vspace{0.5em}

    \item $\sigma = \Theta(n^{\beta})$ for any fixed constant $\beta \in (0,1]$. We use the FID of Liang \& Zhou~\cite[Thm.~1]{LiangZ25}, supporting access and predecessor in $O(\log{\log{r}})$ time for any $t = O(1)$. Therefore, our encoding solves queries in $O(\log{\log{n}})$ worst-case time.
\end{itemize}

Combining these encodings with the one of Sec.~\ref{sec:enc_const_alphabet} proves Thm.~\ref{th:opt_encodings_all}. 

%Because of the page limit, details about construction time and space are deferred to Appendix~\ref{app:construction}.

\section{A Cell-Probe Space-Time Indexing Lower Bound}\label{sec:cell_probe_lb}
%In the cell-probe model, the memory is a sequence of cells, each storing a $w$-bit string, where $w$ is a parameter of the model; we assume $w = \Theta(\log n)$. The cost of an algorithm is measured solely by the number of memory probes, each consisting of reading or writing $w$ bits, while all other computation is free.
%Since our problem is static, we assume that every probe corresponds to a read. The memory accesses of any algorithm can therefore be represented as a decision tree in which each node is labeled by a memory address and corresponds to a probe at that location. Each node has $2^w$ outgoing edges, one for each possible $w$-bit value that can be read, and the cost of an algorithm in this model is thus given by the depth of its corresponding decision tree. 
In the cell-probe model, the memory is a sequence of cells, each storing a $w$-bit string, where $w$ is a parameter of the model; we assume $w = \Theta(\log n)$. The cost of an algorithm is measured by the number of memory probes, each consisting of reading or writing $w$ bits, while all other computation is free. For static problems, every probe corresponds to a read. The memory accesses of any algorithm can therefore be represented as a decision tree in which each node is labeled by a memory address and corresponds to a probe at that location. Each node has $2^w$ outgoing edges, each labeled with one of the possible $w$-bit values that can be read. We assume w.l.o.g. that the edges are sorted by increasing label. The cost of an algorithm is then given by the depth of its corresponding decision tree.

We give a cell-probe lower bound for the indexing version of \textsc{All-Distinct} using the technique of~\cite{Golynski07,BrodalDR12}. In their framework, a stronger assumption is made. Probes into the index are assumed to be free, and the cost of a computation is measured only by the number of probes to the underlying input array. To this end, we begin by restricting ourselves to the following class of instances:

\begin{definition}\label{def:cell_probe_seq}
%Let $n$ and $b$ be two integers such that $1 < b \le n$ and $b \mid n$. The set $\mathcal{C}$ is defined as the set of arrays $A$ of size $n$ over an alphabet of size $b+1$, where for each $0\le i< n/b$ the block $A[ib, (i+1)b-1]$ contains exactly two occurrences of the symbol $1$. The remaining $b-2$ positions are filled by setting any other position $k$ to the value $(k \bmod b) + 2$.

Let $n$ and $b$ be two integers such that $1 < b \le n$ and $b \mid n$. $\mathcal{C}$ is the set of arrays $A$ of size $n$ over an alphabet of size $b+1$, where for each $0\le i< n/b$ the block $A[ib, (i+1)b-1]$ contains exactly two occurrences of the symbol $1$. The remaining positions are filled by setting any other index $k$ to $(k \bmod b) + 2$.
\end{definition}

We now observe that, given a data structure $\mathcal{D}$ built on an unknown instance $I \in \mathcal{C}$, one can reconstruct any block of $I$ by issuing at most $2(\log b+1)$ \textsc{All-Distinct} queries. To see this, consider a generic block $A[ib, (i+1)b-1]$, and the sequence of answers to \textsc{All-Distinct} queries in which the left endpoint is fixed at $ib$ while the right endpoint ranges from $ib+1$ to $(i+1)b-1$. The resulting sequence of answers is a non-increasing binary sequence, transitioning from $1$ to $0$ as soon as the unique duplicated pair is included in the queried range. It follows that the rightmost position $r_i$ of the duplicated pair within the block can be identified via binary search using at most $\log b + 1$ \textsc{All-Distinct} queries to $\mathcal{D}$. By a symmetric argument, the sequence of answers to \textsc{All-Distinct} queries in which the right endpoint is fixed at $r_i$ and the left endpoint ranges from $r_i - 1$ to $ib$ is again a non-increasing binary sequence, transitioning from $1$ to $0$ as soon as the leftmost position $l_i$ of the duplicated pair is included in the queried range. Therefore, by issuing at most $2(\log b + 1)$ \textsc{All-Distinct} queries to $\mathcal{D}$, one can uniquely reconstruct any of the $n/b$ blocks of $I$.

Suppose that an indexing data structure $\mathcal{D}$ uses $n/b$ bits, thus the number of distinct such structures is $2^{n/b}$. Since $|\mathcal{C}| = \binom{b}{2}^{n/b}$, by the pigeonhole principle, there is at least one data structure shared by at least $(\binom{b}{2}/2)^{n/b} = (b(b-1)/4)^{n/b}$ inputs from $\mathcal{C}$. Denote by $\mathcal{C}_{\mathcal{D}} \subseteq \mathcal{C}$ the set of such inputs.

\begin{figure}[ht]
    \centering
    \includegraphics[width=\linewidth]{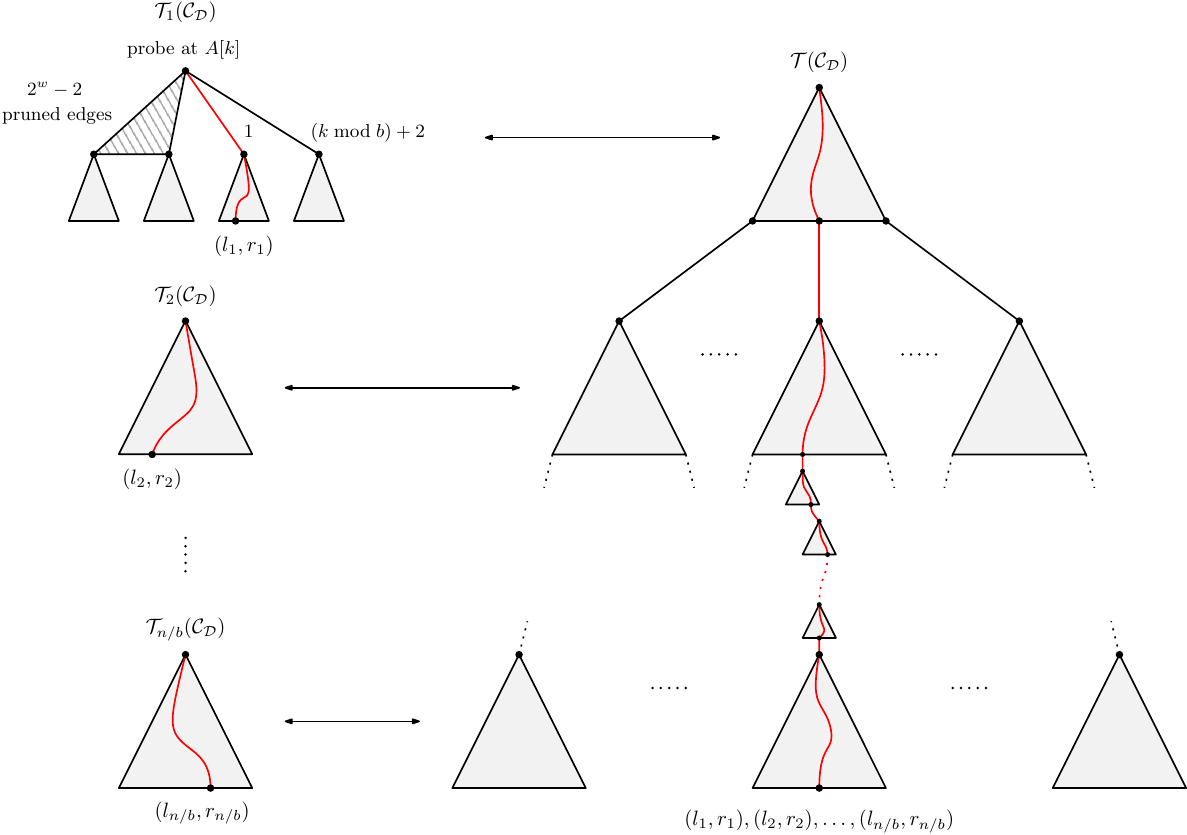}
    \caption{Construction of the composed decision tree $\mathcal{T}(\mathcal{C}_{\mathcal{D}})$ used in the lower bound of Thm.~\ref{th:cell_probe_lb}. Left: each per-block decision tree $\mathcal{T}_{i}(\mathcal{C}_{\mathcal{D}})$ reconstructs the positions $(l_i,r_i)$ of the unique duplicated pair in the $i$-th block using at most $2(\log b+1)$ \textsc{All-Distinct} queries. Each internal node is labeled by a probe at some position $A[k]$. By Def.~\ref{def:cell_probe_seq}, only two outcomes are possible: $1$ and $(k\bmod b)+2$; the remaining $2^w-2$ outgoing edges are unreachable, and therefore pruned. Right: $\mathcal{T}(\mathcal{C}_{\mathcal{D}})$ is obtained by substituting a copy of $\mathcal{T}_{i+1}(\mathcal{C}_{\mathcal{D}})$ at each leaf of $\mathcal{T}_{i}(\mathcal{C}_{\mathcal{D}})$. Thus, every root-to-leaf path encodes a unique sequence $(l_1,r_1),\ldots,(l_{n/b},r_{n/b})$ of duplicated-pair positions.}
    \label{fig:decison_trees_composition}
\end{figure}

%Suppose that an indexing data structure $\mathcal{D}$ uses $n/b$ bits of space, hence the number of distinct such structures is $2^{n/b}$. Since $|\mathcal{C}| = \binom{b}{2}^{n/b}$, by the pigeonhole principle, there is at least one data structure shared by at least $(\binom{b}{2}/2)^{n/b} = (b(b-1)/4)^{n/b}$ inputs from $\mathcal{C}$. Denote by $\mathcal{C}_{\mathcal{D}} \subseteq \mathcal{C}$ the set of such inputs.

Let us now consider the decision tree $\mathcal{T}_{i}(\mathcal{C_{\mathcal{D}}})$ corresponding to the algorithm described above, which reconstructs the $i$-th block of a given instance $I \in \mathcal{C_{\mathcal{D}}}$ (see Figure~\ref{fig:decison_trees_composition}). 

Although a probe into $A$ can in principle read any of the $2^{w}$ possible values, $\mathcal{T}_i(\mathcal{C}_{\mathcal{D}})$ can be pruned into a binary tree. To see this, consider a probe at $A[k]$, and suppose that previous probes inside the block revealed at most a single $1$ position~\footnote{If both occurrences of $1$ within the block have already been located, the probe value is fully determined, leaving a single outgoing edge.}. Because of Def.~\ref{def:cell_probe_seq}, the probe at position $k$ either returns $1$ or returns the value $(k \bmod b) + 2$, making all other $2^{w}-2$ outgoing edges unreachable. 

We proceed by composing all $n/b$ decision trees into a single binary tree $\mathcal{T}(\mathcal{C_{\mathcal{D}}})$ as follows (see Figure~\ref{fig:decison_trees_composition}). Each leaf of $\mathcal{T}_1(\mathcal{C_{\mathcal{D}}})$ is replaced by a copy of $\mathcal{T}_2(\mathcal{C_{\mathcal{D}}})$; the resulting leaves are then replaced by copies of $\mathcal{T}_3(\mathcal{C_{\mathcal{D}}})$, and so on. Every leaf of the resulting tree $\mathcal{T}(\mathcal{C_{\mathcal{D}}})$ is then labeled with its unique sequence $(l_1, r_1), (l_2, r_2), \ldots, (l_{n/b}, r_{n/b})$ of positions of the duplicated symbols across all $n/b$ blocks. Finally, we apply a last pruning step to $\mathcal{T}(\mathcal{C_{\mathcal{D}}})$, in which unreachable nodes are removed (e.g., any node reached via a probe returning a value inconsistent with a previous probe at the same position) and nodes with a single outgoing edge are subsequently collapsed. Notice that, as a result, no repeated probes are performed.

Now, since every two distinct input arrays in $\mathcal{C_{\mathcal{D}}}$, correspond to different leaves in $\mathcal{T}(\mathcal{C_{\mathcal{D}}})$, the number of leaves $L$ of $\mathcal{T}(\mathcal{C_{\mathcal{D}}})$ is at least $(b(b-1)/4)^{n/b}$, i.e., the minimum number of inputs in $\mathcal{C}_{D}$. Furthermore, every root-to-leaf path in $\mathcal{T}(\mathcal{C}_{D})$ can be interpreted as a binary string of length $d$, where $d$ is the depth of $\mathcal{T}(\mathcal{C}_{D})$. This is because we either traverse a left edge corresponding to a probe in which a $1$ is read, or we traverse the other edge. Because of our construction, each of these paths contains at most $2n/b$ left edges. By padding these sequences with further $0$s and $1$s to a total length of $d+2n/b$, we can ensure that they contain exactly $2n/b$ $1$s. Since there are at most $\binom{d+2n/b}{2n/b}$ such sequences it follows that $L \ge (b(b-1)/4)^{n/b}$ and $L \le \binom{d+2n/b}{2n/b}$. We now relate these upper and lower bounds on $L$ to derive a bound on $d$. Taking logarithms and applying the standard inequality $\log\binom{m}{k} \le k\log(\frac{me}{k})$, we obtain:
\begin{align*}
\frac{n}{b}\log{\left(\frac{b(b-1)}{4}\right)} \le \log\binom{d+2n/b}{2n/b} \le \frac{2n}{b}\log\left(\frac{(d+2n/b)e}{2n/b}\right)
\end{align*}
%Dividing both sides by $n/b$, exponentiating in base $2$, and rearranging the left-hand-side gives:
Which we further simplify into the following:
\begin{align*}
\frac{b^2}{4}\left(1-\frac{1}{b}\right) \le \left(\frac{(d+2n/b)e}{2n/b}\right)^{2}
\end{align*}

%Since there are at most $\binom{d+2n/b}{2n/b}$ such sequences, we obtained the following lower and upper bounds on $L$:

%\begin{equation*}
%    \left(\frac{b(b-1)}{4}\right)^{n/b} \le L \le \binom{d+2n/b}{2n/b}
%\end{equation*}

Therefore, $d \ge \frac{n}{e}\sqrt{1-1/b} - 2n/b = \Omega(n)$. Since at most $2(n/b)(\log{b}+1)$ \textsc{All-Distinct} queries are used to reconstruct any instance from $\mathcal{C}_\mathcal{D}$, by the pigeonhole principle, there is at least one such instance for which $\Omega(d/((n/b)\log{b})) = \Omega(b/\log{b})$ probes into the underlying array are needed, thus proving Thm.~\ref{th:cell_probe_lb}.

%Therefore, $d \ge \frac{n}{e}\sqrt{1-1/b} - 2n/b$, and thus $d = \Omega(n)$. Since at most $2(n/b)(\log{b}+1)$ \textsc{All-Distinct} queries are needed to reconstruct any instance from $\mathcal{C}_\mathcal{D}$, by the pigeonhole principle, there is at least one such instance for which $\Omega(d/(n/b)\log{b}) = \Omega(b/\log{b})$ probes into the underlying array are needed. Overall, this proves Thm.~\ref{th:cell_probe_lb}.
%and shows that the index in~\cite {FischerL26} does not yield the optimal space-time tradeoff.

\section{A More Space-Efficient Index}\label{sec:new_index}
The index of Fischer \& Lari~\cite{FischerL26} partitions the array into $n/b$ consecutive blocks of size $b \ge 1$, storing samples $S_{P_A}$ and $S_{N_A}$ of $P_{A}$ and $N_{A}$ at the end and beginning of each block, respectively. Both sequences are accessed in $O(1)$ time and stored in $O((n/b)\log{b})$ bits (see~\cite{FischerL26} for the details). Given a query $\textsc{All-Distinct}(l,r)$, if $r - l + 1 \le b$, they directly solve an $O(b)$-sized element distinctness instance in $T_{dist}(b,\sigma)$ time. Otherwise, let $[l',r']$ be the maximal part of $[l,r]$ that perfectly overlaps with blocks; they verify that it does not introduce duplicates in $[l,r]$ by checking whether $S_{P}[r'] < l$ and $S_{N}[l']>r$. If this is the case, an element in $[l,l'-1]$ can still occur in $[r'+1,r]$ (and vice versa), hence the query is answered by solving an $O(b)$-sized element distinctness instance on $[l,l'-1] \cup [r'+1,r]$ in $T_{dist}(b,\sigma)$ time.

We improve on this with a simple observation. Since the query time is dominated by $T_{dist}(b,\sigma)$, handling the aligned part in $O(1)$ time is unnecessary for any $b = \omega(1)$. We proceed by showing that slowing down this step saves an $O(\log{b})$ factor in space.

For every block $0 \le i < n/b$, we store in $D_{P_A}$ the index of the block containing $p = P_{A}[(i+1)b-1]$, and in $C_{P_A}$ the block index in $A[p, (i+1)b-1]$ containing the duplicate of $A[p-1]$. $D_{N_A}$ and $C_{N_A}$ are defined similarly for $N_A$, but sampling at the first position of each block. Given a query $\textsc{All-Distinct}(l,r)$, suppose $r-l+1 > b$ and that $[l',r']$ is the maximal part of $[l,r]$ overlapping with blocks, we proceed as follows:

\begin{itemize}
    \item $D_{P_A}[\lfloor r'/b \rfloor] > \lfloor l/b \rfloor$: the query is immediately false.
    \item $D_{P_A}[\lfloor r'/b \rfloor] = \lfloor l/b \rfloor$: let $c = C_{P_A}[\lfloor r'/b \rfloor]$. If $c > \lfloor l/b \rfloor$, solve an $O(b)$-sized element distinctness instance on $A[l,l'-1] \cup A[cb, (c+1)b-1]$; if $c = \lfloor l/b \rfloor$, solve it on $A[l,l'-1]$ alone.
    %\item $D_{P_A}[\lfloor r'/b \rfloor] = \lfloor l/b \rfloor$: let $c = C_{P_A}[\lfloor r'/b \rfloor]$, we solve an $O(b)$-sized element distinctness instance on $A[l,l'-1] \cup A[cb, (c+1)b-1]$.
    \item $D_{P_A}[\lfloor r'/b \rfloor] < \lfloor l/b \rfloor$: $[l,r']$ contains no duplicates.
\end{itemize}

The same check is done for $[l', r]$ using $D_{N_A}$ and $C_{N_A}$. In case of success, we answer the query by verifying that $A[l, l'-1] \cup A[r'+1, r]$ does not contain any duplicates in $O(T_{dist}(b,\sigma))$ time. 

Because of their definition, $D_{P_A}$ and $D_{N_A}$ are non-decreasing sequences. The same property holds for $C_{P_A}$ and $C_{N_A}$, but it is less obvious. Focus on $C_{P_A}$, as the same idea applies to $C_{N_A}$. Suppose there are two blocks $i < j$ such that $C_{P_A}[i] > C_{P_A}[j]$. Because of Def~\ref{def:p_n_array}, $P_{A}[(j+1)b-1] \ge P_{A}[(i+1)b-1]$, thus the only possibility is to have $P_{A}[(j+1)b-1] = P_{A}[(i+1)b-1]$, otherwise, we would have $C_{P_A}[i] < C_{P_A}[j]$. However, since $C_{P_A}[i] > C_{P_A}[j]$ and $p = P_{A}[(j+1)b-1] = P_{A}[(i+1)b-1]$, there is another occurrence of $A[p-1]$ before block $C_{P_A}[i]$ and after $p$, which is a contradiction. Because of their monotonicity, we can encode each sequence as a bit vector by storing from left to right the negated unary representation of the differences between consecutive elements. For instance, consider $D_{P_A}$, and let $\Delta_{D_{P}}[i] = D_{P}[i]-D_{P}[i-1]$ for $1\le i< n/b$ with $\Delta_{D_P}[0] = 0$. We store the bit vector $B_{D_P} = 0^{\Delta_{D_P}[0]}10^{\Delta_{D_P}[1]}1\dots0^{\Delta_{D_P}[n/b-1]}1$. Now, the size of $B_{D_P}$ is at most $2(n/b)-1$ bits, since it contains $n/b$ ones and the sum of zeros telescopes to $D_{P_A}[n/b-1] \le n/b-1$. Since $D_{P_A}[i] = \textsc{select}_{1}(B_{D_P}, i+1)-i$, every element of $D_{P_A}$ is accessed in $O(1)$ time by adding $o(n/b)$ bits on top of $B_{D_P}$~\cite{Navarro16book}. Overall, the four sequences $D_{P_A}$, $C_{P_A}$, $D_{N_A}$, and $C_{N_A}$ are stored in $8n/b+o(n/b)$ bits and accessed in $O(1)$ time. Considering the lower bound of Thm.~\ref{th:cell_probe_lb}, using an optimal comparison-based sorting algorithm, the query time is $O(b\log{b})$, which is only an $O(\log^{2}{b})$ factor away from the optimal time within our space usage.

\section{Conclusions}\label{sec:conclusions}
We studied the data structure version of the element distinctness problem introduced by Fischer \& Lari~\cite{FischerL26}, complementing their results in both the encoding and indexing models. In the encoding model, we derived a lower bound for uniformly random inputs and designed average-case optimal encodings matching this bound up to lower-order terms. Our results cover the full alphabet range $1\leq\sigma\leq n$, with worst-case query times ranging from $O(1)$ to $O(\log\log n)$; for $\sigma = \omega(1)$, we additionally give an encoding with $O(1)$ expected query time. In the indexing model, we proved a cell-probe space-time tradeoff lower bound of $\Omega(b/\log b)$ on the query time for any index using $n/b$ bits, and gave a simple index whose query time is within an $O(\log^2 b)$ factor of this bound. A natural direction for future work is to close this gap, either by strengthening the lower bound or by improving the index.

\bibliographystyle{splncs04}
\bibliography{bib}

@article{AlonO94,
  author       = {Noga Alon and
                  Alon Orlitsky},
  title        = {A lower bound on the expected length of one-to-one codes},
  journal      = {{IEEE} Trans. Inf. Theory},
  volume       = {40},
  number       = {5},
  pages        = {1670--1672},
  year         = {1994},
  doi          = {10.1109/18.333891},
}

@article{DBLP:journals/mst/ChanDLMW14,
  author       = {Timothy M. Chan and
                  Stephane Durocher and
                  Kasper Green Larsen and
                  Jason Morrison and
                  Bryan T. Wilkinson},
  title        = {Linear-Space Data Structures for Range Mode Query in Arrays},
  journal      = {Theory Comput. Syst.},
  volume       = {55},
  number       = {4},
  pages        = {719--741},
  year         = {2014},
  url          = {https://doi.org/10.1007/s00224-013-9455-2},
  doi          = {10.1007/S00224-013-9455-2},
}

@inproceedings{DBLP:conf/isaac/BrodalJ09,
  author       = {Gerth St{\o}lting Brodal and
                  Allan Gr{\o}nlund J{\o}rgensen},
  title        = {Data Structures for Range Median Queries},
  booktitle    = {Algorithms and Computation, 20th International Symposium, {ISAAC}
                  2009, Honolulu, Hawaii, USA, December 16-18, 2009. Proceedings},
  series       = {Lecture Notes in Computer Science},
  volume       = {5878},
  pages        = {822--831},
  publisher    = {Springer},
  year         = {2009},
  url          = {https://doi.org/10.1007/978-3-642-10631-6\_83},
  doi          = {10.1007/978-3-642-10631-6\_83}
}

@inproceedings{DBLP:conf/wea/CarmonaL26,
  author       = {Gabriel Carmona and
                  Filippo Lari},
  title        = {Compressing Highly Repetitive Binary Trees with an Application to
                  Range Minimum Queries},
  booktitle    = {24th International Symposium on Experimental Algorithms, {SEA} 2026,
                  Copenhagen, Denmark, June 22-24, 2026},
  series       = {LIPIcs},
  volume       = {371},
  pages        = {10:1--10:20},
  publisher    = {Schloss Dagstuhl - Leibniz-Zentrum f{\"{u}}r Informatik},
  year         = {2026},
  url          = {https://doi.org/10.4230/LIPIcs.SEA.2026.10},
  doi          = {10.4230/LIPICS.SEA.2026.10},
}

@misc{VanDerHoogRR2025-ElementDistinctness,
      title={Tight Better-Than-Worst-Case Bounds for Element Distinctness and Set Intersection}, 
      author={Ivor van der Hoog and Eva Rotenberg and Daniel Rutschmann},
      year={2025},
      eprint={2511.02954},
      archivePrefix={arXiv},
      primaryClass={cs.DS},
      url={https://arxiv.org/abs/2511.02954}, 
}

@inproceedings{bille2025dynamic,
  title={Dynamic range minimum queries on the ultra-wide word RAM},
  author={Bille, Philip and G{\o}rtz, Inge Li and P{\'e}rez L{\'o}pez, M{\'a}ximo and Stordalen, Tord},
  booktitle={International Conference on Current Trends in Theory and Practice of Computer Science},
  pages={122--135},
  year={2025},
  organization={Springer}
}

@article{DBLP:journals/ipl/PetersenG09,
  author       = {Holger Petersen and
                  Szymon Grabowski},
  title        = {Range mode and range median queries in constant time and sub-quadratic
                  space},
  journal      = {Inf. Process. Lett.},
  volume       = {109},
  number       = {4},
  pages        = {225--228},
  year         = {2009},
  url          = {https://doi.org/10.1016/j.ipl.2008.10.007},
  doi          = {10.1016/J.IPL.2008.10.007},
}

@incollection{DBLP:reference/algo/Fischer16,
  author    = {Johannes Fischer},
  title     = {Compressed Range Minimum Queries},
  booktitle = {Encyclopedia of Algorithms},
  publisher = {Springer},
  pages     = {379--382},
  year      = {2016},
  url       = {https://doi.org/10.1007/978-1-4939-2864-4\_640},
  doi       = {10.1007/978-1-4939-2864-4\_640},
}

@inproceedings{DBLP:conf/sofsem/Petersen08a,
  author       = {Holger Petersen},
  editor       = {Viliam Geffert and
                  Juhani Karhum{\"{a}}ki and
                  Alberto Bertoni and
                  Bart Preneel and
                  Pavol N{\'{a}}vrat and
                  M{\'{a}}ria Bielikov{\'{a}}},
  title        = {Improved Bounds for Range Mode and Range Median Queries},
  booktitle    = {{SOFSEM} 2008: Theory and Practice of Computer Science, 34th Conference
                  on Current Trends in Theory and Practice of Computer Science, Nov{\'{y}}
                  Smokovec, Slovakia, January 19-25, 2008, Proceedings},
  series       = {Lecture Notes in Computer Science},
  volume       = {4910},
  pages        = {418--423},
  publisher    = {Springer},
  year         = {2008},
  url          = {https://doi.org/10.1007/978-3-540-77566-9\_36},
  doi          = {10.1007/978-3-540-77566-9\_36},
}

@article{PETERSEN2009225,
title = {Range mode and range median queries in constant time and sub-quadratic space},
journal = {Information Processing Letters},
volume = {109},
number = {4},
pages = {225-228},
year = {2009},
issn = {0020-0190},
doi = {https://doi.org/10.1016/j.ipl.2008.10.007},
author = {Holger Petersen and Szymon Grabowski},
}

@article{DBLP:journals/tcs/BrodalGJS11,
  author       = {Gerth St{\o}lting Brodal and
                  Beat Gfeller and
                  Allan Gr{\o}nlund J{\o}rgensen and
                  Peter Sanders},
  title        = {Towards optimal range medians},
  journal      = {Theor. Comput. Sci.},
  volume       = {412},
  number       = {24},
  pages        = {2588--2601},
  year         = {2011},
  url          = {https://doi.org/10.1016/j.tcs.2010.05.003},
  doi          = {10.1016/J.TCS.2010.05.003},
}

@article{DBLP:journals/njc/KrizancMS05,
  author       = {Danny Krizanc and
                  Pat Morin and
                  Michiel H. M. Smid},
  title        = {Range Mode and Range Median Queries on Lists and Trees},
  journal      = {Nord. J. Comput.},
  volume       = {12},
  number       = {1},
  pages        = {1--17},
  year         = {2005},
}

@article{DBLP:journals/siamcomp/FischerH11,
  author    = {Johannes Fischer and
               Volker Heun},
  title     = {Space-Efficient Preprocessing Schemes for Range Minimum Queries on
               Static Arrays},
  journal   = {{SIAM} J. Comput.},
  volume    = {40},
  number    = {2},
  pages     = {465--492},
  year      = {2011},
  url       = {https://doi.org/10.1137/090779759},
  doi       = {10.1137/090779759},
  bibsource = {dblp computer science bibliography, https://dblp.org}
}

@article{DBLP:journals/siamcomp/BerkmanV93,
  author       = {Omer Berkman and
                  Uzi Vishkin},
  title        = {Recursive Star-Tree Parallel Data Structure},
  journal      = {{SIAM} J. Comput.},
  volume       = {22},
  number       = {2},
  pages        = {221--242},
  year         = {1993},
  url          = {https://doi.org/10.1137/0222017},
  doi          = {10.1137/0222017},
}

@inproceedings{bender2000lca,
  title={The LCA problem revisited},
  author={Bender, Michael A and Farach-Colton, Martin},
  booktitle={LATIN 2000: Theoretical Informatics: 4th Latin American Symposium, Punta del Este, Uruguay, April 10-14, 2000 Proceedings 4},
  pages={88--94},
  year={2000},
  organization={Springer}
}

@inproceedings{DBLP:conf/wea/BaumstarkGHL17,
  author       = {Niklas Baumstark and
                  Simon Gog and
                  Tobias Heuer and
                  Julian Labeit},
  title        = {Practical Range Minimum Queries Revisited},
  booktitle    = {16th International Symposium on Experimental Algorithms, {SEA} 2017,
                  June 21-23, 2017, London, {UK}},
  series       = {LIPIcs},
  volume       = {75},
  pages        = {12:1--12:16},
  publisher    = {Schloss Dagstuhl - Leibniz-Zentrum f{\"{u}}r Informatik},
  year         = {2017},
  url          = {https://doi.org/10.4230/LIPIcs.SEA.2017.12},
  doi          = {10.4230/LIPICS.SEA.2017.12},
  bibsource    = {dblp computer science bibliography, https://dblp.org}
}

@article{DBLP:journals/jda/FerradaN17,
  author       = {H{\'{e}}ctor Ferrada and
                  Gonzalo Navarro},
  title        = {Improved Range Minimum Queries},
  journal      = {J. Discrete Algorithms},
  volume       = {43},
  pages        = {72--80},
  year         = {2017},
  url          = {https://doi.org/10.1016/j.jda.2016.09.002},
  doi          = {10.1016/J.JDA.2016.09.002},
  bibsource    = {dblp computer science bibliography, https://dblp.org}
}

@inproceedings{DBLP:conf/stoc/GabowBT84,
  author       = {Harold N. Gabow and
                  Jon Louis Bentley and
                  Robert Endre Tarjan},
  title        = {Scaling and Related Techniques for Geometry Problems},
  booktitle    = {Proceedings of the 16th Annual {ACM} Symposium on Theory of Computing,
                  April 30 - May 2, 1984, Washington, DC, {USA}},
  pages        = {135--143},
  publisher    = {{ACM}},
  year         = {1984},
  url          = {https://doi.org/10.1145/800057.808675},
  doi          = {10.1145/800057.808675},
}

@inproceedings{LiangZ25,
  author       = {Jingxun Liang and
                  Renfei Zhou},
  title        = {Optimal Static Fully Indexable Dictionaries},
  booktitle    = {52nd International Colloquium on Automata, Languages, and Programming,
                  {ICALP} 2025, Aarhus, Denmark, July 8-11, 2025},
  series       = {LIPIcs},
  volume       = {334},
  pages        = {114:1--114:20},
  publisher    = {Schloss Dagstuhl - Leibniz-Zentrum f{\"{u}}r Informatik},
  year         = {2025},
  doi          = {10.4230/LIPICS.ICALP.2025.114},
}

@inproceedings{FerraginaL25,
  author       = {Paolo Ferragina and
                  Filippo Lari},
  editor       = {Paola Bonizzoni and
                  Veli M{\"{a}}kinen},
  title        = {{FL-RMQ:} {A} Learned Approach to Range Minimum Queries},
  booktitle    = {36th Annual Symposium on Combinatorial Pattern Matching, {CPM} 2025,
                  June 17-19, 2025, Milan, Italy},
  series       = {LIPIcs},
  volume       = {331},
  pages        = {7:1--7:23},
  publisher    = {Schloss Dagstuhl - Leibniz-Zentrum f{\"{u}}r Informatik},
  year         = {2025},
  url          = {https://doi.org/10.4230/LIPIcs.CPM.2025.7},
  doi          = {10.4230/LIPICS.CPM.2025.7},
}

@inproceedings{AnderssonT00,
  author       = {Arne Andersson and
                  Mikkel Thorup},
  editor       = {F. Frances Yao and
                  Eugene M. Luks},
  title        = {Tight(er) worst-case bounds on dynamic searching and priority queues},
  booktitle    = {Proceedings of the Thirty-Second Annual {ACM} Symposium on Theory
                  of Computing, May 21-23, 2000, Portland, OR, {USA}},
  pages        = {335--342},
  publisher    = {{ACM}},
  year         = {2000},
  url          = {https://doi.org/10.1145/335305.335344},
  doi          = {10.1145/335305.335344},
}

@InProceedings{FischerL26,
  author =	{Fischer, Johannes and Lari, Filippo},
  title =	{{Indexing and Encoding Arrays for Element Distinctness Queries}},
  booktitle =	{37th Annual Symposium on Combinatorial Pattern Matching (CPM 2026)},
  pages =	{9:1--9:17},
  series =	{Leibniz International Proceedings in Informatics (LIPIcs)},
  ISBN =	{978-3-95977-420-8},
  ISSN =	{1868-8969},
  year =	{2026},
  volume =	{369},
  editor =	{Bille, Philip and Prezza, Nicola},
  publisher =	{Schloss Dagstuhl -- Leibniz-Zentrum f{\"u}r Informatik},
  address =	{Dagstuhl, Germany},
  URN =		{urn:nbn:de:0030-drops-259350},
  doi =		{10.4230/LIPIcs.CPM.2026.9},
}

@article{GuptaHSV07,
  author       = {Ankur Gupta and
                  Wing{-}Kai Hon and
                  Rahul Shah and
                  Jeffrey Scott Vitter},
  title        = {Compressed data structures: Dictionaries and data-aware measures},
  journal      = {Theor. Comput. Sci.},
  volume       = {387},
  number       = {3},
  pages        = {313--331},
  year         = {2007},
  url          = {https://doi.org/10.1016/j.tcs.2007.07.042},
  doi          = {10.1016/J.TCS.2007.07.042},
}

@book{Navarro16book,
  author       = {Gonzalo Navarro},
  title        = {Compact Data Structures - {A} Practical Approach},
  publisher    = {Cambridge University Press},
  year         = {2016},
  isbn         = {978-1-10-715238-0},
}

@article{BrodalDR12,
  author       = {Gerth St{\o}lting Brodal and
                  Pooya Davoodi and
                  S. Srinivasa Rao},
  title        = {On Space Efficient Two Dimensional Range Minimum Data Structures},
  journal      = {Algorithmica},
  volume       = {63},
  number       = {4},
  pages        = {815--830},
  year         = {2012},
  url          = {https://doi.org/10.1007/s00453-011-9499-0},
  doi          = {10.1007/S00453-011-9499-0},
  bibsource    = {dblp computer science bibliography, https://dblp.org}
}

@article{Flajolet95,
title = {On Ramanujan's Q-function},
journal = {Journal of Computational and Applied Mathematics},
volume = {58},
number = {1},
pages = {103-116},
year = {1995},
issn = {0377-0427},
doi = {https://doi.org/10.1016/0377-0427(93)E0258-N},
author = {Philippe Flajolet and Peter J. Grabner and Peter Kirschenhofer and Helmut Prodinger},
}

@article{DBLP:journals/tcs/JoS26,
  author       = {Seungbum Jo and
                  Srinivasa Rao Satti},
  title        = {Encodings for range minimum queries over bounded alphabets},
  journal      = {Theor. Comput. Sci.},
  volume       = {1070},
  pages        = {115824},
  year         = {2026},
  url          = {https://doi.org/10.1016/j.tcs.2026.115824},
  doi          = {10.1016/J.TCS.2026.115824},
}

@book{cover1999elements,
  title={Elements of information theory},
  author       = {Thomas M. Cover and
                  Joy A. Thomas},
  year={1999},
  publisher={John Wiley \& Sons}
}

@article{Golynski07,
  author       = {Alexander Golynski},
  title        = {Optimal lower bounds for rank and select indexes},
  journal      = {Theor. Comput. Sci.},
  volume       = {387},
  number       = {3},
  pages        = {348--359},
  year         = {2007},
  url          = {https://doi.org/10.1016/j.tcs.2007.07.041},
  doi          = {10.1016/J.TCS.2007.07.041},
}

@book{Knuth97,
  author       = {Donald Ervin Knuth},
  title        = {The art of computer programming, Volume {I:} Fundamental Algorithms,
                  3rd Edition},
  publisher    = {Addison-Wesley},
  year         = {1997},
  isbn         = {0201896834},
}

@inproceedings{Patrascu08,
  author       = {Mihai P{\u{a}}tra{\c{s}}cu},
  title        = {Succincter},
  booktitle    = {49th Annual {IEEE} Symposium on Foundations of Computer Science, {FOCS}
                  2008, October 25-28, 2008, Philadelphia, PA, {USA}},
  pages        = {305--313},
  publisher    = {{IEEE} Computer Society},
  year         = {2008},
  doi          = {10.1109/FOCS.2008.83},
}

\appendix

\section{Missing Proofs from Section~\ref{sec:entropy}}\label{app:missing_proofs_entropy}

For ease of reference, we restate every lemma and proceed with its proof.

\begin{lemma}\label{lem:p_i_q_app}
Let $A$ be an array of $n$ elements drawn independently and uniformly at random from $[1,\sigma]$, and let $p_k := \prod_{j=1}^{k}(1 - j/\sigma)$ for $k \ge 0$, with $p_0 = 1$ and $p_k = 0$ for $k \ge \sigma$. For every $1 \le i \le n-1$ and $0 \le q \le i-1$,
\begin{equation*}
    \mathbb{P}(P_{A}[i-1] = q) \;=\;
    \begin{cases}
        p_{i-1}                 & q = 0,\; i \le \sigma, \\[2pt]
        \tfrac{(i-q)}{\sigma}p_{(i-q)-1}     & q \ge 1,\; 1 \le i-q \le \sigma, \\[2pt]
        0                       & \text{otherwise.}
    \end{cases}
\end{equation*}
\end{lemma}

\begin{proof}
The event $P_{A}[i-1] = q$ requires that $A[q,i-1]$ are distinct and, if $q \ge 1$, that $A[q-1]$ collides with some element of $A[q,i-1]$.

For $q = 0$ (which requires $i \le \sigma$, else $A[0\,..\,i-1]$ cannot be all distinct), the prefix $A[0\,..\,i-1]$ consists of $i$ distinct values, in $\sigma(\sigma-1)\cdots(\sigma-i+1)$ ways out of $\sigma^i$:
\begin{equation*}
    \mathbb{P}(P_{A}[i-1] = 0) \;=\; \frac{\sigma(\sigma-1)\cdots(\sigma-i+1)}{\sigma^i} \;=\; p_{i-1}.
\end{equation*}

For $q \ge 1$ and $1 \le  i - q \le \sigma$, the $i-q$ values $A[q,i-1]$ can be chosen in $\sigma(\sigma-1)\cdots(\sigma-(i-q)+1)$ ways; the element $A[q-1]$ must equal one of these $i-q$ values, accounting for $(i-q)$ other choices; the prefix $A[0,q-2]$ can be any of the $\sigma^{q-1}$ possible configurations. Dividing by $\sigma^i$:
\begin{equation*}
    \mathbb{P}(P_{A}[i-1] = q)
    \;=\; \frac{\sigma(\sigma-1)\cdots(\sigma-(i-q)+1)\,(i-q)}{\sigma^{(i-q)+1}}
    \;=\; \frac{(i-q)}{\sigma}\,p_{(i-q)-1}
\end{equation*}

\end{proof}

\begin{lemma}\label{lem:p_i_cond_app}
Let $A$ be an array of $n$ elements drawn independently and uniformly at random from $[1,\sigma]$. For every $1 \le i \le n-1$ and every $0 \le q \le i-1$ with $i - q \le \sigma$,
\begin{equation*}
    \mathbb{P}(P_{A}[i] = p \mid P_{A}[i-1] = q) \;=\;
    \begin{cases}
        1 - (i - q)/\sigma & p = q, \\[2pt]
        1/\sigma           & q < p \le i, \\[2pt]
        0                  & \text{otherwise.}
    \end{cases}
\end{equation*}
\end{lemma}

\begin{proof}
Given $P_{A}[i-1] = q$, the suffix $A[q,i-1]$ consists of $i - q$ distinct elements of $[1,\sigma]$, and $A[i]$ is independent and
uniformly distributed on $[1,\sigma]$.

If $A[i]$ coincides with $A[j]$ for some $q \le j \le i-1$, then $P_{A}[i] = j + 1$, which lies in $\{q+1, \dots, i\}$; each of the $i - q$ matches has probability $1/\sigma$ and corresponds to a distinct value of $P_{A}[i]$. If instead $A[i]$ differs from every $A[j]$ with $j \in \{q, \dots, i-1\}$, then $A[q,i]$ are all distinct and $P_{A}[i] = q$; this happens with probability $1 - (i-q)/\sigma$.
\end{proof}

\begin{lemma}\label{lem:identities_app}
Let $p_k := \prod_{j=1}^{k}(1 - j/\sigma)$ for $k \ge 0$, with $p_0 = 1$ and $p_k = 0$ for $k \ge \sigma$. The following identities hold for every $\sigma \ge 1$:
\begin{align}
    p_{k-1} - p_k &\;=\; \tfrac{k}{\sigma}\,p_{k-1},
        \qquad 1 \le k \le \sigma, \label{eq:id_basic_app} \\[4pt]
    \sum_{k=0}^{\sigma-1} k\,p_k &\;=\; \sigma - Q(\sigma). \label{eq:id_kpk_app}
\end{align}
\end{lemma}

\begin{proof}
Eq.~\ref{eq:id_basic_app} follows from $p_k = p_{k-1}\,(1 - k/\sigma)$, i.e., the recurrence implicit in the
definition of $p_k$. Multiplying both sides by $\sigma$ gives
$(k+1)\,p_k = \sigma\,(p_k - p_{k+1})$ after a shift of index. Summing
from $k = 0$ to $k = \sigma - 1$ telescopes the right-hand side to
$\sigma\,(p_0 - p_\sigma) = \sigma$, hence $\sum_{k=0}^{\sigma-1}(k+1)p_k = \sigma$, giving:
\begin{equation*}
    \sum_{k=0}^{\sigma-1} k\,p_k
    \;=\; \sum_{k=0}^{\sigma-1} (k+1)\,p_k \;-\; \sum_{k=0}^{\sigma-1} p_k
    \;=\; \sigma \;-\; Q(\sigma),
\end{equation*}
which is Eq.~\ref{eq:id_kpk_app}.
\end{proof}

\section{Selecting the FIDs in Section~\ref{sec:enc_gen_alphabet}}\label{app:missing_proofs_enc}

\subsection{Polylogarithmic Alphabet Sizes ($\sigma = O(\polylog{n})$)}\label{subsec:polylog_alphabet}
We implement the FIDs of our encoding in Sec.~\ref{sec:enc_gen_alphabet} using Pătrașcu's data structure~\cite{Patrascu08}, whose redundancy is $n/(\log{n}/t)^{t} + \Tilde{O}(n^{3/4})$ bits with $O(t)$ (worst-case) query time. Noting that the redundancy does not depend on the random variable $r$, we proceed as follows. When $t = O(1)$, $\Tilde{O}(n^{3/4}) = o(n/\log^{t}{n})$, the redundancy is just $O(n/\log^{t}{n})$ bits. To match the lower bound of Thm.~\ref{th:p_sequence_entropy_sigma}, we need $O(n/\log^{t}{n}) = o((n/\sqrt{\sigma})\log{\sigma})$, which holds whenever $\sigma = O(\log^{2t}{n})$. Therefore, since we can set $t$ to any positive constant, the expected space usage of our encoding is optimal for any $\sigma = O(\polylog{n})$, and the worst-case query time is $O(1)$.

Lastly, it is possible to show that the redundancy of Pătrașcu's FID is dominated by the $\Tilde{O}(n^{3/4})$ term when $t = \Omega(\log{n}/\log{\log{n}})$. In this case, $\Tilde{O}(n^{3/4}) = o((n/\sqrt{\sigma})\log{\sigma})$ when $\sigma = o(\sqrt{n})$. Therefore, it is outside the current regime of $\sigma$, and in Appendix~\ref{subsec:poly_alphabet} we also design an even better solution for such a case.

\subsection{Sub-polynomial Alphabet Sizes ($\sigma = n^{o(1)}$)}\label{subsec:subpoly_alphabet}

We use Gupta et al.~\cite[Thm.~2]{GuptaHSV07} FID to implement our encoding in Sec.~\ref{sec:enc_gen_alphabet}. Its space usage is $gap+O(r\log{(n/r)}/\log{r}+r\log{\log{(n/r)}})$ bits, where $gap$ is the well-known data-aware measure obtained by summing the logarithms of the distances (i.e., gaps) between consecutive values in the underlying sequence. By Jensen's inequality $gap \le \log{\binom{n}{r}}+O(r)$, therefore it fits the analysis of Sec.~\ref{sec:enc_gen_alphabet}. Focusing on the redundancy, we notice that the first term $r\log{(n/r)}/\log{r}$ is not always a concave function of $r$, however, it has an inflection point (passing from convex to concave) at $r = e^2$, which can be computed through derivatives. Let $f(r) = r\log{(n/r)}/\log{r}$ and $\ell(r)$ be the tangent line of $f(r)$ at $r = e^2$. We decompose $f(r)$ as $f(r) = f_1(r) + f_2(r)$, where $f_1(r) = \min\{f(r),\ell(r)\}$ and $f_2(r) = \max\{f(r)-\ell(r),0\}$. By linearity of the expected value, $\mathbb{E}[f(r)] = \mathbb{E}[f_1(r)] + \mathbb{E}[f_2(r)]$. $f_2(r)$ is non-zero for $r < e^2$, and thus $\mathbb{E}[f_2(r)] = O(\log{n})$. $f_1(r)$ is concave on the whole range of $r$, by Jensen's inequality $\mathbb{E}[f_1(r)] \le f_1(\mathbb{E}[r])$. By Lemma~\ref{lem:average_runs} and for large enough $n$, $\mathbb{E}[r] > e^2$ and thus $\mathbb{E}[f_1(r)] \le f(\mathbb{E}[r]) = O(\frac{n\log{\sigma}}{\sqrt{\sigma}\log{(n/\sqrt{\sigma})}})$. Overall, in the current regime of $\sigma = n^{o(1)}$, $\mathbb{E}[f(r)] = o((n/\sqrt{\sigma})\log{\sigma})$.

The other $O(r\log{\log{(n/r)}})$ term is instead always a concave function of $r$, thus by Jensen's inequality and Lemma~\ref{lem:average_runs}, its expected value is $O((n/\sqrt{\sigma})\log{\log{\sigma}})$. Overall, for the current regime of $\sigma = n^{o(1)}$, the expected value of the redundancy is always $o((n/\sqrt{\sigma})\log{\sigma})$, and thus the encoding matches the lower bound of Thm.~\ref{th:p_sequence_entropy_sigma}. Considering the query time, Gupta et al. FID supports access in $O(\log{\log{r}})$ time and predecessor in $AT(n,r)$ time, where $AT(n,r)$ is Andersson and Thorup's predecessor data structure time bound~\cite{AnderssonT00,GuptaHSV07}, and is defined as follows:

\begin{equation*}
    AT(n,r) = O\left(\min\left\{\sqrt{\frac{\log{r}}{\log{\log{r}}}},\, \frac{\log{\log{n}}}{\log{\log{\log{n}}}}\cdot\log{\log{r}},\,\log{\log{r}}+\frac{\log{r}}{\log{\log{n}}}\right\}\right)
\end{equation*}

$AT(n,r)$ is then bounded by $o(\log^{2}{\log{n}})$ in the worst-case, and thus our encoding inherits the same query time.

Lastly, we show why Pătrașcu's FID is not well-suited in this case. Noting that its redundancy is not a function of $r$ and focusing on super-constant query time, $\Tilde{O}(n^{3/4}) = o(n/(\log{n}/t)^{t})$ when $t = O(\log{n}/\log{\log{n}})$, so the redundancy is $O(n/(\log n/t)^t)$ bits. Matching the lower bound of Thm.~\ref{th:p_sequence_entropy_sigma} requires $n/(\log n/t)^t = o((n/\sqrt{\sigma})\log{\sigma})$, that is, $(\log{n}/t)^{t} > \frac{\sqrt{\sigma}}{\log{\sigma}}$, which, taking logarithms, translates to
\begin{equation}\label{eq:ineq_app}
    t(\log{\log{n}}-\log{t}) > \tfrac{1}{2}\log{\sigma}-\log{\log{\sigma}}.
\end{equation}
To simplify the calculations, since $\tfrac{1}{2}\log{\sigma}-\log{\log{\sigma}} < \tfrac{1}{2}\log{\sigma}$, it suffices to enforce the stronger condition

\begin{equation}\label{eq:ineq_str_app}
    t(\log{\log{n}}-\log{t}) \ge \tfrac{1}{2}\log{\sigma}.
\end{equation}

Now, we focus on the largest asymptotic value of the LHS of Eq.~\ref{eq:ineq_str_app}, since comparing against the maximum gives us the largest bound on the alphabet size. To this end, we notice that such a value is attained at $t = O(\log{n}/\log{\log{n}})$, and is equal to $O(\tfrac{\log{n}\,\log{\log{\log{n}}}}{\log{\log{n}}})$. Therefore, $\log{\sigma} = O(\tfrac{\log{n}\log{\log{\log{n}}}}{\log{\log{n}}})$, and thus $\sigma = n^{O\left(\tfrac{\log{\log{\log{n}}}}{\log{\log{n}}}\right)}$, which is only a strict subset of $n^{o(1)}$. Moreover, any $\sigma$ in that range, guarantees that it is always possible to choose a $t = O(\log{n}/\log{\log{n}})$ such that $n/(\log n/t)^t = o((n/\sqrt{\sigma})\log{\sigma})$, and therefore, in this regime, the expected space usage of our encoding would match the lower bound of Thm.~\ref{th:p_sequence_entropy_sigma} with a slower worst-case query time of $O(\log{n}/\log{\log{n}})$.

\subsection{Polynomial Alphabet Sizes ($\sigma = n^{\beta}$)}\label{subsec:poly_alphabet}

When $\sigma = n^{\beta}$ for a fixed constant $\beta \in (0,1]$, we can only use Pătrașcu's FID in the framework of Sec.~\ref{sec:enc_gen_alphabet} when $\beta < 1/2$ (see Appendix~\ref{subsec:polylog_alphabet}). On the other hand, we can still use Gupta et al.'s data structure, thus obtaining an average-case optimal encoding with $o(\log^{2}\log{n})$ worst-case time (see Appendix~\ref{subsec:subpoly_alphabet}). Here we show that, under the current alphabet regime, the worst-case query time can be improved. Let $\mu = \mathbb{E}[r]$, then by Lemma~\ref{lem:average_runs}, $\mu = \sqrt{\tfrac{\pi}{2}}n^{1-\beta/2}(1+o(1))$ and thus $n = \Theta(\mu^{\delta})$ for some fixed constant $\delta \in (1,2]$. Since the universe size is a polynomial in the number of elements, we can use the recently introduced FID of Liang \& Zhou~\cite[Thm.~1]{LiangZ25}, which supports queries in $O(\log{\log{r}})$ time for any $t = O(1)$ with a redundancy of $O(r/(\log^{t}{r}))$ bits. Noting that the redundancy is in turn bounded by $r$, its expected value is then $O(n/\sqrt{\sigma})$ bits, which is $o((n/\sqrt{\sigma})\log{\sigma})$ for any $\sigma = \omega(1)$. Overall, following the framework of Sec.~\ref{sec:enc_gen_alphabet}, the expected space usage of our encoding matches the lower bound of Thm.~\ref{th:p_sequence_entropy_sigma}, with $O(\log{\log{n}})$ worst-case query time.

\end{document}